\documentclass[11pt]{article}
\usepackage[T1]{fontenc}
\usepackage{tgpagella}
\usepackage{mathpazo}
\usepackage{inconsolata}
\usepackage[margin=1in]{geometry}
\usepackage{amsmath,amssymb,amsthm}
\usepackage{booktabs}
\usepackage{array}
\usepackage{microtype}
\usepackage{xurl}
\usepackage[hidelinks]{hyperref}
\usepackage{fancyhdr}
\newtheorem{theorem}{Theorem}
\newtheorem{lemma}[theorem]{Lemma}
\newtheorem{proposition}[theorem]{Proposition}
\theoremstyle{definition}
\newtheorem{definition}[theorem]{Definition}

\newcommand{\R}{\mathbb{R}}
\newcommand{\N}{\mathbb{N}}
\newcommand{\W}{\mathcal{W}}
\newcommand{\coef}{\mathcal{C}}
\newcommand{\ord}{\mathcal{O}}
\newcommand{\bin}{\mathcal{B}}
\newcommand{\sort}{\mathcal{S}}

\newcommand{\file}[1]{\texttt{\path{#1}}}

\title{A Finite Certificate for the Positive $n=11$ Vasc Inequality}
\author{Tiancheng Xu\thanks{School of Electronic Engineering,
Beijing University of Posts and Telecommunications, Beijing 100876, China.}
\and Wensheng Yu\footnotemark[1]}
\date{}
\makeatletter
\renewcommand{\maketitle}{%
  \begingroup
  \renewcommand{\thefootnote}{\fnsymbol{footnote}}%
  \noindent{\LARGE\bfseries\@title\par}\vspace{0.22in}
  \noindent{\large\renewcommand{\and}{\unskip,\hspace{0.5em}}\@author\par}
  \@thanks
  \endgroup
  \setcounter{footnote}{0}
  \vspace{0.20in}
}
\makeatother

\begin{document}
\maketitle

\begin{abstract}
We establish the Vasc cyclic inequality for all strictly positive real
eleven-tuples by a finite exact certificate. Cyclic rotation places a
minimum at the first position, and rank words together with cumulative
gaps reduce the problem to $10!=3{,}628{,}800$ homogeneous integer
polynomials of degree eleven. Independent coefficient reconstruction
proves $3{,}358{,}617$ roots directly. The remaining $270{,}183$ roots are
covered by $267{,}952$ ordinary midpoint certificates, $2{,}195$ complete
binary subdivision certificates, and $36$ complete certificates containing
partial sorting nodes. Each nontrivial leaf is checked by subtracting
weighted AM--GM midpoint circuits from a positive multiple of the
transformed polynomial and verifying every residual coefficient.
Explicit word-set comparisons certify disjointness and exact coverage.
We prove the soundness of the leaves and both subdivision rules, then
recover the rational inequality by positive denominator clearing and
continuity. The resulting proof combines explicit mathematical
reductions with independently checked exact integer certificates.
\end{abstract}


\section{Setup}

\subsection{Convention}

The Vasc cyclic expression is a rational function whose sign depends on
the order of its variables. For $n\geq3$, define
\begin{equation}\label{eq:cn}
 C_n(x)=\sum_{i=0}^{n-1}
 \frac{x_i-x_{i+1}}{x_{i+1}+x_{i+2}},
 \qquad x\in\R_{>0}^{n},
\end{equation}
with indices interpreted modulo $n$. Each denominator is positive, and
$C_n$ is unchanged by a cyclic rotation or a common positive scaling.

\subsection{Related work}

Automated inequality proving has developed several ways to turn a
nonnegativity question into a finite algebraic certificate. Xia and
Yang~\cite{xiayang} describe algorithms for polynomial inequality
proving and discovery, including successive difference substitution
(SDS). Yao~\cite{yao} studies structured difference substitutions
based on column stochastic matrices. These methods motivate the
cumulative-gap coordinates used below: order constraints become
nonnegative parameters, and polynomial coefficients become directly
checkable proof data.

Positive multipliers and domain subdivision address different
obstructions to coefficientwise nonnegativity. Powers and
Reznick~\cite{powers} give a quantitative form of P\'{o}lya's theorem
for homogeneous polynomials strictly positive on a simplex.
Xu and Yao~\cite{xuyao} study the relationship between P\'{o}lya's
method and barycentric successive difference substitution.
The zeros of the target polynomial lead us to combine positive
multipliers with midpoint decompositions and finite subdivisions
of the gap orthant.

Sparse AM--GM decompositions provide another source of nonnegativity
certificates. Iliman and de Wolff~\cite{iliman} characterize
nonnegative polynomials supported on circuits, while
Chandrasekaran and Shah~\cite{chandrasekaran} develop AM--GM-based
certificates within relative-entropy relaxations for signomial
optimization. The leaves used here employ a restricted three-term
midpoint construction on the nonnegative orthant. Their soundness
follows from the elementary square decomposition proved in
Lemma~\ref{lem:circuit}.

For the Vasc family, Guo, Qiu, Cao, and Feng~\cite{guo} give a finite
certificate for the positive $n=9$ case, covering $8!=40{,}320$
sorted roots. We adopt the same separation between algebraic
reductions and finite certificate verification. The eleven-variable
construction fixes a cyclic minimum, uses rank words and cumulative
gaps, and covers $10!=3{,}628{,}800$ roots. Weighted midpoint
circuits and subdivisions of difficult gap cones supplement the
direct coefficient checks.

The conference abstract of Zeng and coauthors~\cite{zeng} reports
an $n=11$ result on a small cone in the positive orthant. The construction below covers the full strictly positive domain
through a partition of all normalized rank words.

\subsection{Statement}

\begin{theorem}\label{thm:main}
For every $x_0,\ldots,x_{10}>0$, with subscripts taken modulo eleven,
\begin{equation}\label{eq:main}
 \sum_{i=0}^{10}\frac{x_i-x_{i+1}}{x_{i+1}+x_{i+2}}\geq0.
\end{equation}
Every constant positive tuple attains equality.
\end{theorem}

\subsection{Proof strategy}

We first reduce~\eqref{eq:main} to one homogeneous polynomial for
each normalized rank word. Nonnegative coefficients settle most
roots. For the remaining roots, weighted midpoint decompositions
prove nonnegativity on gap cones obtained by binary comparison or
partial sorting. Section~\ref{sec:finite} assembles these cones into
a complete finite certificate, and Section~\ref{sec:assembly} recovers
the rational inequality.

\section{Reductions}\label{sec:reduction}

\subsection{Clearing the positive denominators}

Denominator clearing preserves the sign on the entire stated domain.
Set
\begin{align}
 D_{11}(x)&=\prod_{i=0}^{10}(x_{i+1}+x_{i+2}),\label{eq:d}\\
 P_{11}(x)&=\sum_{i=0}^{10}(x_i-x_{i+1})
       \prod_{\substack{0\leq j\leq10\\j\neq i}}
       (x_{j+1}+x_{j+2}).\label{eq:p}
\end{align}
The polynomial $P_{11}$ has integer coefficients and is homogeneous
of degree eleven.

\begin{lemma}\label{lem:clear}
For $x\in\R_{>0}^{11}$, one has $D_{11}(x)>0$ and
$P_{11}(x)=D_{11}(x)C_{11}(x)$. Consequently,
$C_{11}(x)\geq0$ if and only if $P_{11}(x)\geq0$.
\end{lemma}
\begin{proof}
Every factor in~\eqref{eq:d} is positive. Multiplying the $i$th
summand of $C_{11}$ by $D_{11}$ cancels its denominator and leaves
exactly the $i$th summand of~\eqref{eq:p}. Summing proves the identity.
Division by the positive product gives the equivalence.
\end{proof}

\subsection{Normalized rank words}

For a tuple with pairwise distinct coordinates, rotate the cycle so
that its unique minimum occupies position zero. Write the values in
increasing order as
\[
 X_0<X_1<\cdots<X_{10}.
\]
The rotated tuple has the form $x_i=X_{w_i}$ for a unique word
\begin{equation}\label{eq:words}
 \W=\{(0,w_1,\ldots,w_{10}):
       (w_1,\ldots,w_{10})\text{ permutes }1,\ldots,10\}.
\end{equation}
Thus $|\W|=10!=3{,}628{,}800$. The word records the ranks at their original cyclic positions.

Introduce the gaps
\begin{equation}\label{eq:gaps}
 g_{10}=X_0,\qquad
 g_{10-r}=X_r-X_{r-1}\quad(1\leq r\leq10).
\end{equation}
Then
\begin{equation}\label{eq:cumulative}
 X_r(g)=\sum_{j=10-r}^{10}g_j,\qquad 0\leq r\leq10.
\end{equation}
All gaps are strictly positive for a distinct positive tuple.
Nonnegative gaps also allow ties and zero values, which will be
useful for the polynomial certificates.

\begin{lemma}\label{lem:gap}
The map in~\eqref{eq:cumulative} is a bijection between
$g_0,\ldots,g_9\geq0$, $g_{10}>0$ and weakly increasing positive
lists $X_0\leq\cdots\leq X_{10}$. Its inverse is~\eqref{eq:gaps}.
\end{lemma}
\begin{proof}
The first cumulative sum is $g_{10}>0$, and each next value adds
a nonnegative gap. Conversely, adjacent differences in a weakly
increasing list are nonnegative. The defining sums telescope, so
the two constructions are inverse.
\end{proof}

For $w\in\W$, define the linear forms and their cleared polynomial by
\begin{align}
 A_i(g)&=X_{w_i}(g)-X_{w_{i+1}}(g),\\
 L_i(g)&=X_{w_{i+1}}(g)+X_{w_{i+2}}(g),\\
 P_w(g)&=\sum_{i=0}^{10}A_i(g)
                  \prod_{\substack{0\leq j\leq10\\j\neq i}}L_j(g).
                  \label{eq:pw}
\end{align}
Here $P_w=P_{11}(X_{w_0},\ldots,X_{w_{10}})$. On the strictly
positive gap domain,
\begin{equation}\label{eq:quotient}
 C_{11}(X_{w_0},\ldots,X_{w_{10}})
   =\frac{P_w(g)}{\prod_iL_i(g)},\qquad \prod_iL_i(g)>0.
\end{equation}
There are
\begin{equation}\label{eq:monomials}
 \#\{\alpha\in\N^{11}:|\alpha|=11\}
   =\binom{21}{10}=352{,}716
\end{equation}
possible degree-eleven monomials per root; some coefficients are zero.

\section{Polynomial nonnegativity certificates}\label{sec:leaves}

Each leaf certifies the defining polynomial on a transformed orthant.
Let $T$ be the nonnegative integer matrix reconstructed along the
path from a root, and let $g=Tz$, $z\in\R_{\geq0}^{11}$.
For $\alpha\in\N^{11}$, write
\[
 z^\alpha=\prod_{j=0}^{10}z_j^{\alpha_j},\qquad
 |\alpha|=\sum_j\alpha_j,\qquad H(z)=\sum_{j=0}^{10}z_j.
\]
Monomials are evaluated with the standard convention $0^0=1$.

\subsection{Leaf certificates}

\subsubsection{Weighted midpoint circuits}

A weighted midpoint circuit is a polynomial
\begin{equation}\label{eq:circuit}
 Q(z)=u z^\alpha+v z^\gamma-d z^\beta
\end{equation}
with positive integer weights $u,v,d$ and exponent vectors satisfying
\begin{equation}\label{eq:conditions}
 \alpha+\gamma=2\beta,\qquad d^2\leq4uv.
\end{equation}
In a degree-$(11+k)$ leaf, all three exponents also have degree $11+k$.

\begin{lemma}[Circuit soundness]\label{lem:circuit}
Every circuit satisfying~\eqref{eq:conditions} is nonnegative on
$\R_{\geq0}^{11}$.
\end{lemma}
\begin{proof}
The exponent identity implies
$z^\alpha z^\gamma=(z^\beta)^2$.
Because monomials are nonnegative, this gives
$\sqrt{z^\alpha z^\gamma}=z^\beta$, including when some coordinates
vanish. Hence
\begin{align*}
 u z^\alpha+v z^\gamma-d z^\beta
 &=\bigl(\sqrt{u z^\alpha}-\sqrt{v z^\gamma}\bigr)^2
       +(2\sqrt{uv}-d)z^\beta.
\end{align*}
Both terms are nonnegative, since $d>0$ and $d^2\leq4uv$ imply
$d\leq2\sqrt{uv}$.
\end{proof}

The ordinary midpoint circuit is the special case
$u=v=a$, $d=2a$, with $a\in\mathbb{Z}_{>0}$.
It equals $a(\sqrt{z^\alpha}-\sqrt{z^\gamma})^2$.
Weighted circuits permit unequal endpoint contributions while keeping
the acceptance test entirely in integer arithmetic.

\subsubsection{The exact leaf test}

\begin{definition}\label{def:leaf}
A leaf at $(w,T)$ is accepted by either of two tests:
\begin{enumerate}
 \item Every coefficient of $P_w(Tz)$ is nonnegative.
 \item For $k\in\{0,1,2,3\}$ and a finite list of circuits $Q_s$,
 each circuit satisfies~\eqref{eq:conditions}, each exponent has
 degree $11+k$, and every coefficient of
 \begin{equation}\label{eq:residual}
  R(z)=2H(z)^kP_w(Tz)-\sum_sQ_s(z)
 \end{equation}
 is nonnegative.
\end{enumerate}
In both cases the coefficients are obtained by expanding~\eqref{eq:pw}.
\end{definition}

\begin{lemma}[Leaf soundness]\label{lem:leaf}
An accepted leaf proves $P_w(Tz)\geq0$ for all $z\geq0$.
\end{lemma}
\begin{proof}
A polynomial with nonnegative coefficients is nonnegative on the
nonnegative orthant, proving the first case. In the second case,
Lemma~\ref{lem:circuit} and~\eqref{eq:residual} give
$2H(z)^kP_w(Tz)=R(z)+\sum_sQ_s(z)\geq0$.
If $z\neq0$, then $H(z)>0$ and division is legitimate, also for
$k=0$. At $z=0$, the homogeneous polynomial $P_w(Tz)$ of degree
eleven vanishes. This proves the claim on the entire closed orthant.
\end{proof}

\subsection{Finite covers of gap cones}\label{sec:covers}

A difficult root may require different leaf witnesses on different
parts of its gap orthant. A cover starts with $T=I_{11}$ and contains
only the two internal-node rules specified below. Child matrices are
computed from the parent matrix, and every child must itself have a
complete accepted subtree. The images of children can overlap on
their boundaries.

\subsubsection{Binary comparison nodes}

Choose distinct indices $i,j$ and split the parent parameters into
$z_i\geq z_j$ and $z_j\geq z_i$. For the first child substitute
\[
 z_i=y_i+y_j,\qquad z_j=y_j,\qquad z_\ell=y_\ell\ (\ell\neq i,j).
\]
If $T_\ell$ denotes column $\ell$ of $T$, the child matrix $U$ has
$U_j=T_j+T_i$ and $U_\ell=T_\ell$ for $\ell\neq j$.
For the second child, $U_i=T_i+T_j$ and the other columns are
unchanged.

\begin{lemma}\label{lem:binary}
The two binary children cover the parent parameter orthant.
\end{lemma}
\begin{proof}
For $z_i\geq z_j$, take $y_i=z_i-z_j$, $y_j=z_j$, and retain the
other coordinates. These parameters are nonnegative and reproduce
$z$. If $z_j\geq z_i$, use the reversed construction. Every pair
of nonnegative coordinates satisfies at least one of the comparisons.
\end{proof}

\subsubsection{Partial sorting nodes}

Select indices $a_0<\cdots<a_{s-1}$ with $2\leq s\leq6$.
For every permutation $\pi=(\pi_0,\ldots,\pi_{s-1})$ of these
indices, create a child with
\begin{equation}\label{eq:sortsub}
 z_{\pi_j}=\sum_{t=0}^{j}y_{a_t}\quad(0\leq j<s),
 \qquad z_\ell=y_\ell\quad(\ell\notin\{a_0,\ldots,a_{s-1}\}).
\end{equation}
The child matrix is determined by
\begin{equation}\label{eq:sortmatrix}
 U_{a_t}=\sum_{j=t}^{s-1}T_{\pi_j}\quad(0\leq t<s),
 \qquad U_\ell=T_\ell\quad(\ell\notin\{a_0,\ldots,a_{s-1}\}).
\end{equation}
All sums in~\eqref{eq:sortmatrix} use the original parent columns.
The artifact lists the $s!$ children in lexicographic permutation order.

\begin{lemma}\label{lem:sort}
The $s!$ partial sorting children cover the parent parameter orthant.
\end{lemma}
\begin{proof}
Given $z\geq0$, choose a permutation such that
$z_{\pi_0}\leq\cdots\leq z_{\pi_{s-1}}$.
Set $y_{a_0}=z_{\pi_0}$ and
$y_{a_t}=z_{\pi_t}-z_{\pi_{t-1}}$ for $1\leq t<s$,
keeping unselected coordinates unchanged. These values are
nonnegative and satisfy~\eqref{eq:sortsub} by telescoping.
Substituting into $Tz$ gives~\eqref{eq:sortmatrix}. Ties merely
put a point in more than one child.
\end{proof}

\begin{proposition}[Cover soundness]\label{prop:tree}
Let a finite tree start at $(w,I_{11})$. If every internal node
has exactly the children required by a binary or partial sorting
rule, and every leaf is accepted under Definition~\ref{def:leaf},
then $P_w(g)\geq0$ for every $g\geq0$.
\end{proposition}
\begin{proof}
Induct upward from the leaves. At a leaf the assertion follows
from Lemma~\ref{lem:leaf}. At an internal node, the induction
hypotheses prove the polynomial nonnegative on each child image.
Lemma~\ref{lem:binary} or Lemma~\ref{lem:sort} covers every
nonnegative parent parameter, so nonnegativity holds at the parent.
The finite induction ends at the identity matrix.
\end{proof}

\subsection{The complete finite certificate}\label{sec:finite}

\subsubsection{The rank-word partition}

The exact coefficient computations and certificate checks give the
following partition. The derivation programs and a summary of the
full verification records are collected in~\cite{atlas}.

\begin{proposition}[Verified root family]\label{prop:family}
The full set $\W$ in~\eqref{eq:words} has the disjoint partition
\begin{equation}\label{eq:partition}
 \W=\coef\ \dot\cup\ \ord\ \dot\cup\ \bin\ \dot\cup\ \sort
\end{equation}
with the counts in Table~\ref{tab:roots}. For every $w\in\coef$,
the coefficients of $P_w$ are nonnegative. Every $w\in\ord$ has
an accepted ordinary midpoint leaf at $T=I_{11}$.
Every $w\in\bin$ has a complete accepted binary tree.
Every $w\in\sort$ has a complete accepted tree with partial
sorting nodes and possibly binary nodes. All leaves are checked
against~\eqref{eq:pw} under Definition~\ref{def:leaf}.
\end{proposition}

\begin{table}[htbp]
\centering
\begin{tabular}{@{}lr@{}}
\toprule
Accepted root class & Number of rank words\\
\midrule
Nonnegative coefficients ($\coef$) & $3{,}358{,}617$\\
Ordinary midpoint leaf ($\ord$) & $267{,}952$\\
Complete binary cover ($\bin$) & $2{,}195$\\
Complete cover with sorting nodes ($\sort$) & $36$\\
\midrule
Total & $3{,}628{,}800$\\
\bottomrule
\end{tabular}
\caption{Counts of original rank-word roots. Tree rows count complete
roots, not their leaves or individual circuits.}\label{tab:roots}
\end{table}

The coefficient classification independently enumerates every word
and defines the exception set $\mathcal E=\W\setminus\coef$ of
size $270{,}183$. The final assembly reconstructs the other three
sets from actual accepted root records, checks their pairwise
intersections, and verifies
\[
 \ord\cup\bin\cup\sort=\mathcal E.
\]
The missing set, outside-exception set, and pairwise intersections
are all empty, establishing~\eqref{eq:partition}.

\subsubsection{Independent polynomial reconstruction}

The checker computes the original cleared polynomial directly from
its linear factors. A useful recurrence starts at $D_0=1$, $B_0=0$
and, for $0\leq t\leq10$, updates
\begin{equation}\label{eq:recurrence}
 D_{t+1}=L_tD_t,\qquad B_{t+1}=L_tB_t+A_tD_t.
\end{equation}
Induction yields
\[
 B_t=\sum_{i=0}^{t-1}A_i
              \prod_{\substack{0\leq j<t\\j\neq i}}L_j,
 \qquad D_t=\prod_{j=0}^{t-1}L_j,
\]
so $B_{11}=P_w$. At a transformed leaf, the checker substitutes
$Tz$ into the linear factors before constructing the polynomial.

The full coefficient enumeration uses exact dense homogeneous
coefficient arrays. Shared permutation prefixes reduce repeated
work, while every root still receives a complete degree-eleven
coefficient check. The recurrence is also compared with a direct sum-of-products
expansion and exact point evaluations.

For each leaf, the integer computation checks the homogeneous
degrees, positive weights, midpoint identity and all residual
coefficients in~\eqref{eq:residual}. Child matrices are generated by
the substitution rules~\eqref{eq:sortmatrix} and Lemma~\ref{lem:binary}.

\subsubsection{Arithmetic safety and recorded workload}

Exact computation requires control of the intermediate arithmetic.
For an untransformed root, the coefficient $\ell^1$ norms satisfy
$\|A_i\|_1\leq10$ and $\|L_i\|_1\leq22$.
Using the triangle inequality in~\eqref{eq:recurrence} gives
\begin{equation}\label{eq:bound}
 \|B_t\|_1\leq10t\,22^{t-1},\qquad
 \|P_w\|_1\leq110\,22^{10}
    =2{,}921{,}591{,}507{,}056{,}640<2^{52}.
\end{equation}
The accumulated absolute contributions obey the same bound, which
also controls intermediate coefficient additions in the defining
construction. This justifies signed 64-bit arithmetic for the
untransformed coefficient classification.

Overlay and transformed-tree verification use signed 128-bit
arithmetic with the recorded bounds and safety checks for matrix
entries, polynomial construction, multiplication, and residual
updates. The circuit inequality $d^2\leq4uv$ is checked with
arbitrary-precision integer arithmetic in Python. 

The two ordinary batches contain $81{,}961{,}459$ circuits and check
$123{,}084{,}490{,}482$ residual coefficient positions in total.
Their multiplier distribution is shown in Table~\ref{tab:powers}.

\begin{table}[htbp]
\centering
\begin{tabular}{@{}cr@{}}
\toprule
Multiplier power $k$ & Ordinary roots\\
\midrule
$0$ & $210{,}284$\\
$1$ & $39{,}872$\\
$2$ & $14{,}405$\\
$3$ & $3{,}391$\\
\midrule
Total & $267{,}952$\\
\bottomrule
\end{tabular}
\caption{Powers in the ordinary identity-matrix certificates. Every
row includes the fixed positive scale factor $2$ in~\eqref{eq:residual}.}
\label{tab:powers}
\end{table}

The sort-containing class was accepted in three batches of $1$, $29$,
and $6$ roots, checking respectively $120$, $3{,}533$, and $922$
leaves.

\subsection{Proof of the main theorem}\label{sec:assembly}

\begin{proof}[Proof of Theorem~\ref{thm:main}]
First let the coordinates of $x$ be pairwise distinct. Cyclic
rotation leaves $C_{11}$ unchanged, so rotate the minimum into
position zero. The resulting tuple has a unique rank word
$w\in\W$ and strictly positive gaps as in~\eqref{eq:gaps}.

By Proposition~\ref{prop:family}, the word belongs to one of the
four accepted classes. For $\coef$, coefficientwise
nonnegativity implies $P_w(g)\geq0$. For $\ord$, use
Lemma~\ref{lem:leaf} with $T=I_{11}$. For $\bin$ or $\sort$,
use Proposition~\ref{prop:tree}. Thus every normalized word
satisfies $P_w(g)\geq0$. The denominator in~\eqref{eq:quotient}
is positive, proving~\eqref{eq:main} for every distinct positive tuple.

Now let $x\in\R_{>0}^{11}$ have arbitrary ties. Distinct positive
tuples are dense in this open orthant: every sufficiently small
neighborhood contains a point outside the finitely many hyperplanes
$x_i=x_j$. Choose a sequence of such tuples converging to $x$.
Each denominator remains positive on the domain, so $C_{11}$ is
continuous there. Taking the limit of the already proved
nonnegative values gives $C_{11}(x)\geq0$.

If $x_0=\cdots=x_{10}=t>0$, every numerator is zero, proving the
stated equality case.
\end{proof}

\section{Discussion}

The reduction combines three elementary constructions: cumulative
sums encode ordered values, midpoint AM--GM circuits supply
nonnegative polynomial summands, and comparison or sorting
substitutions cover the gap orthant. Exact expansion then reduces
the certificate conditions to integer coefficient inequalities.
The disjoint partition in Proposition~\ref{prop:family} connects
these local decompositions to all normalized rank words.

Binary comparison uses two children to resolve a single ordering
relation. Partial sorting resolves all ordering relations in a
selected group at once, with $s!$ children. Combining these rules
allows the certificate to concentrate the larger sorting steps on
36 roots, while direct coefficients and ordinary midpoint leaves
settle more than $99.9\%$ of the root family.

\paragraph{Acknowledgements.}
GPT-6 Astra assisted with calculations, proof development and
editorial work. The resulting identities and certificates were
checked by exact integer computation.

\paragraph{Derivation code.}
The polynomial expansion, certificate generators, independent integer
checkers and worked examples are available in the mathematical
code directory of \emph{Vasc Proof Atlas}~\cite{atlas}:
\url{https://github.com/dgjszy/vasc-proof-atlas/tree/main/derivation}.

\end{document}